\documentclass[a4paper]{amsart}

\usepackage{amssymb}

\usepackage{orcidlink}

\usepackage{chngcntr, csquotes}
\usepackage{accents}
\usepackage{tikz}
\usetikzlibrary{automata,positioning}
\usepackage{hyperref}
\usepackage{cleveref}
\usepackage{extdash}

\theoremstyle{theorem}
\newtheorem{theorem}{Theorem}[section]
\newtheorem{lemma}[theorem]{Lemma}
\newtheorem{corollary}[theorem]{Corollary}
\newtheorem{definition}[theorem]{Definition}
\newtheorem{proposition}[theorem]{Proposition}

\theoremstyle{definition}
\newtheorem{question}[theorem]{Question}
\newtheorem{remark}{Remark}

\newcommand{\zz}{\mathbb{Z}}
\newcommand{\ii}{\mathcal{I}}
\newcommand{\is}{i}
\newcommand{\pow}{\mathcal{P}}
\newcommand{\grammar}{G}
\newcommand{\ff}{\mathbb{F}}
\newcommand{\nn}{\mathbb{N}}
\newcommand{\pp}{\mathbb{P}} 

\newcommand{\ip}[1]{\left\lfloor #1 \right\rfloor} 
\newcommand{\deriv}{\Longrightarrow^*} 
\newcommand{\cn}{\sqrt[c]{n}} 
\newcommand{\x}[1]{x_{[#1]}} 
\newcommand{\set}[2]{\{ \, #1 \colon #2 \,\}} 

\DeclareMathOperator{\supp}{weight} 
\DeclareMathOperator{\ran}{ran} 

\newcommand{\define}[1]{{\normalfont{\textbf{#1}}}}

\newcommand{\sstar}{\Sigma^*}
\newcommand{\splus}{\Sigma^+}
\newcommand{\xor}{\oplus}

\newcommand{\cD}{A_D}
\newcommand{\cN}{A_{Ne}}
\newcommand{\cNu}{A_{N}}

\newcommand{\mN}{M_{Ne}}
\newcommand{\mNu}{M_{N}}

\newcommand{\sac}{\mathbf{SAC}^0}
\newcommand{\psac}{\mathord{\bigoplus}\mathbf{SAC}^0}
\newcommand{\cosac}{\mathbf{coSAC}^0}
\newcommand{\copsac}{\mathbf{co}\mathord{\bigoplus}\mathbf{SAC}^0}
\newcommand{\prsac}{\mathord{\bigoplus}\mathbf{SAC}^0_r}

\title[Languages of Low-Complexity Words Are Hard to Compute]{Languages of Words of Low Automatic Complexity\newline Are Hard to Compute}

\author[J.~Chen]{Joey Chen}
\address{Department of Mathematics, National University of Singapore, Singapore}
\email{e0389025@u.nus.edu}{}{}

\author[B.~Kjos-Hanssen]{Bj\o{}rn Kjos-Hanssen \orcidlink{0000-0002-6199-1755}}
\address{Department of Mathematics, University of Hawai‘i at M\=anoa, United States of America}
\urladdr{https://math.hawaii.edu/wordpress/bjoern/}
\email{bjoernkh@hawaii.edu}

\author[I.~Koswara]{Ivan Koswara \orcidlink{0000-0002-9311-6840}}
\address{School of Computing, National University of Singapore, Singapore}
\email{ivanak@comp.nus.edu.sg}

\author[L.~Richter]{Linus Richter \orcidlink{0000-0003-0267-1839}}
\address{Department of Mathematics, National University of Singapore, Singapore}
\urladdr{https://linus-richter.github.io}
\email{richter@nus.edu.sg.org}

\author[F.~Stephan]{Frank Stephan \orcidlink{0000-0001-9152-1706}}
\address{Department of Mathematics, National University of Singapore, Singapore and School of Computing, National University of Singapore, Singapore}
\urladdr{https://www.comp.nus.edu.sg/~fstephan/}
\email{fstephan@comp.nus.edu.sg}



\begin{document}

\maketitle

\begin{abstract}
The automatic complexity of a finite word (string) is an analogue for finite automata of Sipser's distinguishing complexity (1983) and was introduced by Shallit and Wang (2001). For a finite alphabet~$\Sigma$ of at least two elements, we consider the non-deterministic automatic complexity given by \emph{exactly}---yet not necessarily \emph{uniquely}---accepting automata: a word $x \in \sstar$ has exact non-deterministic automatic complexity $k \in \nn$ if there exists a non-deterministic automaton of $k$ states which accepts $x$ while rejecting every other word of the same length as $x$, and no automaton of fewer states has this property. Importantly, and in contrast to the classical notion, the witnessing automaton may have multiple paths of computation accepting~$x$. We denote this measure of complexity by $\cN$, and study a class of languages of low~$\cN$-complexity defined as~$L_q = \set{x \in \sstar}{\cN(x) < q|x|}$, which is parameterised by rationals $q \in (0,1/2)$ (generalising a class of sets first studied by Kjos-Hanssen). We show that for every $q \in (0,1/2)$, this class is neither context-free nor recognisable by certain Boolean circuits. In the process, we answer an open question of Kjos-Hanssen quantifying the complexity of~$L_{1/3}$ in terms of Boolean circuits, and also prove the Shannon effect for~$\cN$.
\end{abstract}

\section{Introduction}\label{sec:intro}

Automatic complexity is a notion of complexity of finite words (strings) determined by witnessing automata, first introduced by Shallit and Wang in~\cite{shallitWang} as a Turing computable alternative to Kolmogorov complexity. It is an analogue for finite automata of Sipser's distinguishing complexity \cite{sipser}. Classically, the automatic complexity of a word~$x$ over a finite alphabet $\Sigma$ refers to the cardinality---counted in number of states\footnote{A version of automatic complexity counting the number of \emph{transitions} has been studied by Serna~\cite{serna}; see also Shallit and Breitbart \cite{breitbart}.}---of the smallest deterministic finite automaton which accepts $x$ and rejects every other word of the same length as~$x$~\cite{shallitWang}. The notion as well as  variations of it have proven interesting for multiple reasons. For instance, since automatic complexity is Turing computable, it can be used in the study of computational complexity: the computational complexity of sets of binary words of low automatic complexity has helped prove missing relationships in the Complexity Zoo~\cite{zoo} (see~\cite[Theorem 39]{bjornMax} for an example). Further, the detailed investigation of words in terms of their automatic complexity~\cite{kjosfibo,bjornshift} has shed light on \emph{computable} notions of randomness, which are unavailable from the viewpoint of Kolmogorov complexity~\cite{kolmogorov:three:1965,solomonoff:prelim:1960,martinloef:def:1966}.

In this paper, we study a weakening of a variation of automatic complexity due to Hyde~\cite{hydeMA}, and show that it generates classes of words too complicated to be captured by pushdown automata, nor by certain classes of constant-depth Boolean circuits---both of which are notably computationally more powerful than finite automata. This provides further evidence towards the conjecture that automatic complexity is hard to compute (see e.g. \cite{bjornequiv}).

\subsection{Technical Background}

Fix a finite alphabet $\Sigma$ of at least two elements. In usual Kleene notation, we denote by~$\sstar$ the set of all finite words of elements from $\Sigma$. We denote the empty string by $\varepsilon$, and the set of non-empty words by~$\splus = \sstar \setminus \{ \varepsilon \}$. By an \define{automaton} we always mean a non-deterministic finite automaton, unless otherwise stated. We do not allow $\varepsilon$-transitions.

\begin{definition}
	Let $x \in \sstar$. An automaton $M$ \define{exactly accepts} $x$ if $M$ accepts~$x$, and whenever both~$y \neq x$ and~$|y| = |x|$ then $M$ rejects $y$.
\end{definition}
The pumping lemma shows that this definition is maximally restrictive on the number of words accepted by the witnessing automaton; trying to strengthen the definition by asking for outright uniqueness of the accepted word only leads to trivialities.

\begin{definition}
	The \define{automatic complexity} of $x \in \sstar$ is given by
\[
	\cD(x) = \min \set{k \in \nn}{\text{\normalfont{there exists a DFA of $k$ states which exactly accepts $x$}}}.
\]
\end{definition}
For a reference on contemporary automatic complexity, see e.g.\ the recent \cite{bjornbook}. The subscript~$D$ stands for ``deterministic'', indicating that $\cD(x)$ is determined by the smallest~\textit{D}FA. By definition, it is clear that $\cD{}$ is well-defined, and even computable (for every~$n \in \nn$, there are only finitely many DFAs, and each can be simulated in finite time). However---similar to the unnatural properties of \emph{plain} compared to \emph{prefix-free} Kolmogorov complexity---the measure~$\cD{}$ has the following properties, which may render it undesirable as a natural measure of complexity of words. These were first described in~\cite{hydebjorn}:
\begin{enumerate}
	\item $\cD{}$ is not invariant under natural transformations on strings, such as reversals. For instance, Hyde and Kjos-Hanssen have verified computationally that $\cD(011100) = 4 < 5 = \cD(001110)$.
	\item The DFA witnessing $\cD{}(x)$ often appears unnatural, in the sense that determinism requires $\cD(x)$ to be total: in many cases, an automaton non\Hyphdash{}``deterministically'' witnessing $\cD(x)$ needs to be augmented by an extra state to which every non-accepting path leads.   
\end{enumerate}
To overcome these obstacles, Hyde introduced automatic complexity witnessed by the smallest \emph{non-deterministic} finite automaton (NFA) \cite{hydeMA}.

\begin{definition}
	Let $x \in \sstar$. An automaton $M$ \define{uniquely accepts} $x$ if $M$ exactly accepts $x$ and there is only one path in $M$ which accepts $x$.
\end{definition}
Clearly, every DFA which exactly accepts $x$ also uniquely accepts $x$. For NFAs, however, this is not the case.
An NFA uniquely accepts $x$ if and only if the NFA exactly accepts $x$ and the NFA is unambiguous on~$\Sigma^{|x|}$. Though Hyde \cite{hydeMA} required the NFA to be unambiguous on~$\Sigma^{|x|}$, she noted that the complexity based on NFAs is much more flexible and many words have a smaller complexity in her version than if only DFAs are considered. This led her to:

\begin{definition}\label[definition]{dfn:unn}
	Let $x \in \sstar$. The \define{unique non-deterministic automatic complexity} of~$x$ is
\[
	\cNu(x) = \min \set{k \in \nn}{\text{\normalfont{there exists an NFA of $k$ states which uniquely accepts $x$}}}.
\]	
\end{definition}

\begin{remark}
	We note that this notion is usually called ``non-deterministic automatic complexity''. As we study an ostensibly weaker notion below, we emphasise the additional strength of the notion defined in \Cref{dfn:unn} by adding the attribute ``unique''.
\end{remark}
While it is well-known that NFAs and DFAs recognise exactly the same class of languages---the regular languages (see e.g.\ \cite{Shallit_2008,nerode} for a comprehensive background on automata theory)---the respective notions of automatic complexity differ. The following properties of~$\cNu{}$ have been derived by Hyde and Kjos-Hanssen alongside co-authors, and others. Let $\mNu(x)$ denote both the \define{minimal automaton witnessing $\cNu(x)$} and the directed graph representing it.

\begin{lemma}\label[lemma]{lem:poyt}
	Let $x \in \sstar$.
	\begin{enumerate}
		\item $\cNu(x) \leq (|x|/2) + 1$ follows from exhibiting suitable NFAs {\normalfont{\cite{hydeMA}}}.
		\item $\mNu(x)$ is planar {\normalfont{\cite{bjornPlanar}}}.
	\end{enumerate}

\end{lemma}
Building upon Hyde's work from \cite{hydeMA}, in the present paper we study more closely the notion of automatic complexity induced by a weaker class of machines: the class of exactly but not necessarily uniquely accepting automata.

\begin{definition}\label[definition]{dfn:eacc}
	Let $x \in \sstar$. The \define{non-deterministic automatic complexity} of $x$ is
	\[
		\cN(x) = \min \set{k \in \nn}{\text{\normalfont{there exists an NFA of $k$ states which exactly accepts $x$}}}.
	\]
\end{definition}
Since every NFA which uniquely accepts $x$ also exactly accepts $x$, we immediately see that~$\cN(x) \leq \cNu(x)$. Whether equality holds is still open (\Cref{questionF}). In~\cite{bjornMax}, Kjos-Hanssen investigated the complexity of certain languages induced by~$\cNu{}$ in terms of more complicated models of computation, e.g.\ pushdown automata. In particular, he showed:

\begin{theorem}\label{thm:bjorn}\phantom\\
	\begin{enumerate}
		\item $\set{x \in \{ 0,1,2\}^*}{\cNu(x) \leq |x|/2}$ is not context-free.
		\item $\set{x \in \{ 0,1\}^*}{\cNu(x) \leq |x|/3}$ is not recognised by constant-depth circuits with semi-unbounded fan-in, using Boolean $\land$- and $\lor$-gates.
	\end{enumerate}

\end{theorem}
Results of this type motivate this paper: we investigate the impact of exactness on the behaviour of automatic complexity, which we describe via theorems
akin to \Cref{thm:bjorn}.

\subsection{Our Theorems and the Structure of This Paper}

We investigate the complexity of $\cN{}$ as a function in terms of the complexity of the language of $\cN$-complicated words. Explicitly, we investigate the following class of languages first defined\footnote{In {\cite[Def.\ 17]{bjornMax}}, Kjos-Hanssen has considered the complementary decision problem, given by the decision problem~$q|x| < \cN(x)$. We note that our class $\set{L_q}{q \in (0,1/2)}$ is more general.} by Kjos-Hanssen~\cite{bjornMax}, and prove results on their complexities.

\begin{definition}
	For $q \in (0,1/2)$, define $L_q = \set{x \in \sstar}{\cN(x) < q|x|}.$
\end{definition}
In \Cref{sec:LqProps}, we isolate complexity results on the~$L_q$-sets which follow from a fine-grained investigation of its elements. For instance, in \Cref{prp:Kolm} we isolate an upper bound of the Kolmogorov complexity of words in~$L_q$. This gives a small-to-large result---a theorem about elements which provides information about sets---in the form of \Cref{usecor27}, which shows that the cardinality of $L_q \cap \Sigma^n$ is in $o\left(|\Sigma|^n\right)$. 
This observation also yields a proof of the \define{Shannon effect} for $\cN$:

\theoremstyle{plain}
\newtheorem*{thm:shannon}{\Cref{thm:sh}}
\begin{thm:shannon}
	Let $\cN\left(\Sigma^n\right) = \max_{x \in \Sigma^n} \cN(x)$. For almost every $x \in \Sigma^*$,
	\[
		\cN(x) \geq \cN\left(\Sigma^{|x|}\right) - o\left(\cN\left(\Sigma^{|x|}\right)\right).
	\]
\end{thm:shannon}
In \Cref{sec:Lq}, we demonstrate that pushdown automata are not powerful enough to characterise~$\cN$-complicated words, which the following theorems show.

\newtheorem*{thm:LLq}{\Cref{thm:Lq}}
\begin{thm:LLq}
	For every $q \in (0,1/2)$, the language $L_q$ is not context-free.
\end{thm:LLq}

\newtheorem*{thm:LLLLq}{\Cref{thm:compl}}
\begin{thm:LLLLq}
	For every $q \in (0,1/2)$, the language $\sstar \setminus L_q$ is not context-free.
\end{thm:LLLLq}
In \Cref{sec:circ}, we consider the complexity of $L_q$ in terms of Boolean circuits. To do so, we use two classical types of Boolean circuits---$\sac$, defined in \Cref{ss:c1}, and $\psac$, defined in \Cref{ss:c2}---and apply a counting argument to prove:

\newtheorem*{thm:sss}{\Cref{thm:sac1}}
\begin{thm:sss}
	For $q \in (0,1/2)$ and $|\Sigma| = 2$,~$L_q \not\in \sac$ and $\sstar \setminus L_q \not\in \sac$.
\end{thm:sss}

\newtheorem*{thm:ssss}{\Cref{thm:sac2}}
\begin{thm:ssss}
Let $q \in (0,1/2)$ and $|\Sigma| = p$ for some prime $p$. Then $L_q \not\in \psac$ and~$\sstar \setminus L_q \not\in \psac$.
\end{thm:ssss}
As a special case, we show that $L_{1/3}$ is not $\psac$-recognisable, answering a question of Kjos-Hanssen \cite[p.\ 351]{bjornMax}.

By giving a minor redefinition of $\psac$-recognisability for alphabets of non-prime cardinality, we also prove a partial generalisation of these theorems:

\newtheorem*{thm:ttt}{\Cref{finalThm}}
\begin{thm:ttt}
Let $q \in (0,1/2)$ and $|\Sigma| = r$ for some non-prime $r$. Let $p$ be the smallest prime greater than $r$. Let $\prsac$ denote the class $\psac$ for $r$-cardinality alphabets inside the field of $p$ elements.

Then $L_q \not\in \prsac$ and~$\sstar \setminus L_q \not\in \prsac$.
\end{thm:ttt}

For details on this redefinition, see~\Cref{sec:last} and in particular~\Cref{newDfn}.

\medskip

In \Cref{sec:open}, we conclude this paper by giving a few open questions.

\subsection{Acknowledgements}
	B.~Kjos-Hanssen was partially supported by a grant from the Simons Foundation (\#704836 to Bj{\o}rn Kjos-Hanssen). L.~Richter was fully supported by Singapore Ministry of Education grant MOE-000538-01. F.~Stephan was partially supported by Singapore Ministry of Education grant MOE-000538-01. Parts of this work have appeared in the first author's Bachelor's thesis submitted to the National University of Singapore.

\section[Combinatorial Properties of Sets of Words of Low-Complexity]{Combinatorial Properties of $L_q$}\label{sec:LqProps}

In this section, we derive combinatorial properties of $L_q$ which are needed in the sequel, particularly to prove \Cref{thm:Lq}.
Fix $q \in (0,1/2)$.
Firstly, we show that~$L_q$ satisfies a strong closure property: any word~$x \in \sstar$ can be extended to some word $y \in \sstar$ for which $y \in L_q$.

\begin{proposition}\label[proposition]{prp:longgg}
	Suppose $x \in \sstar$. If $m > |x|/q$ then $x^m \in L_q$.
\end{proposition}

\begin{proof}
	Let $n = |x|$ and suppose $x = x_0\cdots x_{n-1} \in \sstar$. Now build an NFA as follows: there are~$n$ states~$\{ s_0,\ldots, s_{n-1} \}$, with~$s_0$ being both the start and unique accepting state. Transitions are given by~$s_i \xrightarrow{x_i} s_{i+1}$ for~$i < n-1$ and $s_{n-1} \xrightarrow{x_{n-1}} s_{0}$. It is readily seen that this automaton witnesses
	\[
		\cN(x^m) \leq |x| < qm < qm|x| = q|x^m|. \qedhere
	\]
\end{proof}
While the previous proposition employs repetition of words to push the non-deterministic automatic complexity down, in the following lemma we show that spacing out bits of information achieves the same effect. W.l.o.g., assume $0 \in \Sigma$. For notation, if $x = x_0\cdots x_{n-1} \in \sstar$ then define the \define{(Hamming) weight of $x$} by $\supp(x) = |\set{k < n}{x_k \neq 0}|$.

\begin{lemma}[Gap Lemma]\label[proposition]{prp:gapsPrp}
	For every $c \in \nn$ there exists $n \in \nn$ such that if $x \in \Sigma^n$ and $\supp(x) \leq c$ then~$x \in L_q$. 
\end{lemma}
Note that, in the statement above, $n$ depends on $q$, which we fixed at the beginning of this section. Before we give the proof, we need the following number-theoretical lemma, called Bertrand's postulate (for a proof see e.g.\ \cite{murty}).
Let~$\pp$ denote the set of prime numbers.

\begin{lemma}[Bertrand's postulate]
	If $h > 1$ then $\pp \cap (h,2h)$ is non-empty.
\end{lemma}

\begin{proof}[Proof of \Cref{prp:gapsPrp}]
	Fix $c \in \nn$. For each $n \in \nn \setminus \{ 0,1 \}$, we define a finite sequence of primes by~$(p_1(n),\ldots,p_c(n))$ as follows: put $p_1(n) = \min \left(\pp \cap \left(\cn,2\cn\right)\right)$ and
	\begin{align*}
		 &p_{i+1}(n) = \min (\pp \cap (p_i,2p_i)) &\text{for $i = 1,2,\ldots,c-1$}.
	\end{align*}
	Since $n > 1$, Bertrand's Postulate shows that this is well-defined. Now, let
	\[
		Q_i(n) = \left( \frac{1}{p_i(n)} \right) \prod_{j=1}^c p_j(n)
	\]
	Bertrand's postulate alongside a short calculation implies
	\begin{align*}
		p_c(n) &< 2^{c-1}p_1(n) < 2^{c}\cn
	\intertext{and so}
		Q_i(n) &\leq (p_c(n))^{c-1} \leq 2^{c(c-1)}n^{\frac{c-1}{c}}
	\end{align*}
which proves~$Q_i(n) \in O\left(n^{\frac{c-1}{c}}\right)$. This also shows that, in the limit,~$Q_i(n) < n$. Similarly,
\[
	Q_i(n)p_i(n) > (p_1(n))^{c} > (\cn)^c = n
\]
and hence, again in the limit, $Q_i(n) < n < Q_i(n)p_i(n)$.
	For $x \in \Sigma^n$ with $\supp(x) \leq c$, write~$x = w_0 0^{\ell_1}w_1 \cdots 0^{\ell_k}w_k$ for some~$k \leq c$.
	Given~$q \in (0,1/2)$, we now choose~$n < \omega$ sufficiently large so that we may assume the following (we write~$Q_i = Q_i(n)$):
	\begin{itemize}
		\item $|w_i| \leq cQ_1$ for $i = 0,1,\ldots,k$.
		\item $\ell_i \geq Q_1$ for $i = 1,\ldots,k$.
	\end{itemize}
Now, write $\ell_i = a_i Q_i + r_i$ where $0 \leq r_i < Q_i$. Since $Q_i p_i > n$ we must have~$a_i < p_i$; otherwise, we must have~$|\ell_i| > n = |x|$, a contradiction. Hence, consider the automaton $M$ in \Cref{fig:gapAutomaton}.
	
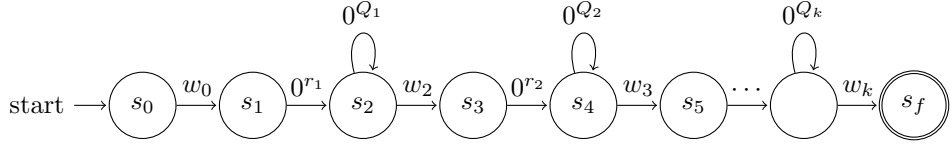
\begin{figure}[htbp]
\begin{center}
	\begin{tikzpicture}[shorten >=1pt,node distance=1.45cm,on grid,auto] 
  		 \node[state,initial] (s_0)  {$s_0$}; 
 		 \node[state] (s_1) [right=of s_0] {$s_1$}; 
   	\node[state] (s_2) [right=of s_1] {$s_2$};
   	\node[state] (s_3) [right=of s_2] {$s_3$}; 
   	\node[state] (s_4) [right=of s_3] {$s_4$}; 
   	\node[state] (s_5) [right=of s_4] {$s_5$};
   	\node[state] (s_6) [right=of s_5] {};
   	\node[state,accepting](s_f) [right=of s_6] {$s_f$};
    	\path[->] 
    		(s_0) edge  node {$w_0$} (s_1)
    		(s_1) edge  node {$0^{r_1}$} (s_2)
    		(s_2) edge [loop above] node {$0^{Q_1}$} ()
    		(s_2) edge  node {$w_2$} (s_3)
    		(s_3) edge  node {$0^{r_2}$} (s_4)
    		(s_4) edge [loop above] node {$0^{Q_2}$} ()
    		(s_4) edge  node {$w_3$} (s_5)
    		(s_5) edge  node {$\cdots$} (s_6)
    		(s_6) edge [loop above] node {$0^{Q_k}$} ()
    		(s_6) edge  node {$w_k$} (s_f);
	\end{tikzpicture}%
\caption{The non-deterministic automaton witnessing the ``gap lemma''.}
\label{fig:gapAutomaton}
\end{center}%
\end{figure}%
We show that $M$ is as required. First, by definition, $M$ accepts $x$. To show exactness, suppose~$y \in \Sigma^n$ and that $M$ accepts $y$. If $x \neq y$, assume w.l.o.g.\ that~$M(y)$ traverses the~$0^{Q_1}$-loop fewer than $a_i$-many times. Since $|y| = n$, $M(y)$ must go through the remaining loops more often to make up for the $Q_1$-deficit. However, the equation $Q_1 = d_2Q_2 + \ldots + d_kQ_k$ has no integer solution, since $p_1$ divides the right-hand side yet not $Q_1$. Thus,~$M$ cannot accept~$y$, as needed. Finally, recall that $r_i, |w_i| \leq cQ_1 \in O\left(n^{\frac{c-1}{c}}\right)$. Thus, for~$n$ large enough, inspection of the automaton~$M$ shows that
\[
	\cN(x) \le 3kcQ_1 \in O\left(n^{\frac{c-1}{c}}\right)
\]
for every~$x \in \Sigma^n$, and thus, eventually,~$\cN(x) < q|x|$.
\end{proof}

As a consequence of our proof, we also obtain the following (we thank the anonymous reviewer for pointing this out).

\begin{corollary}
	For every~$c \in \nn$ and every~$q \in (0,1/2)$ there exists~$n \in \nn$ such that if~$x \in \Sigma^{>n}$ and~$\supp(x) \le c$ then~$x \in L_q$.
\end{corollary}

Our next result studies the small-scale structure of words in~$L_q$. We say~$w$ is a \define{factor of~$x$} if there exist~$u,v \in \sstar$ for which~$x = uwv$; we write~$w \preceq x$. If~$u \in \splus$ or~$v \in \splus$ then~$w$ is a \define{proper factor of~$x$}; we write~$w \prec x$.
Call a non-empty word $w$ a \define{square} if there exists~$v \prec w$ for which~$w = vv$; we write~$w = v^2$.

\begin{proposition}\label[proposition]{prp:2x}
	Let $n \geq 4$ and $x \in L_q \cap \Sigma^n$. There exists a proper factor~$w \prec x$ of length~$|w| \geq \left( \frac{1-2q}{2} \right)\sqrt{n}$ for which there are $u,v \in \splus$ with $|u| = |v| \leq |w|$ and $uw = wv \prec x$. Further, $uwv \prec x$.
\end{proposition}
Note that if $|u| = |v| = |w|$ then the conclusion of \Cref{prp:2x} yields a square. To prove the general case of \Cref{prp:2x}, we again need a classical auxiliary result, in this case due to Lyndon and Sch\"utzenberger \cite{lyndon}.

\begin{theorem}[The First Lyndon-Sch\"utzenberger-Theorem]\label{lynd1}
	Suppose $x,y \in \sstar$. Then $xy = yx$ if and only if there exists $z \in \sstar$ and $k,\ell \in \nn$ for which $x = z^k$ and~$y = z^\ell$.
\end{theorem}
Observe that the First Lyndon-Sch\"utzenberger-Theorem characterises \emph{bordered words}\footnote{For more on bordered words, see e.g.\ \cite{border}. A more general characterisation is given by the Second Lyndon-Sch\"utzenberger-Theorem \ref{thm:ls2}.}---those which have a non-trivial decomposition of the form $uw = wv$---as those generated by powers of a common word $z$. This will be important in the proof of \Cref{prp:2x}.
We also require the following combinatorial lemma.

\begin{lemma}\label[lemma]{lem:combLem}
	Suppose $x \in \Sigma^n$ for some $n \geq 4$. Assume $x \in L_q$, and let $\mN(x)$ be the witnessing automaton with accepting run $(q_0,\ldots,q_n)$. Then
	\[
		|\set{k \in \nn}{(\exists \, i,j)(i < j < k \land q_i = q_j = q_k)}| \geq (1-2q)n.
	\]
\end{lemma}

\begin{proof}
	Consider the list of states $(q_0,\ldots,q_{n})$.
	Since~$q < 1/2$, we have~$2qn < n$. In particular,~$n = 2qn + (1-2q)n$. Hence, by the pigeonhole principle, there exist at least~$(1-2q)n$ indices at which some state is visited a third time.
\end{proof}
We now prove \Cref{prp:2x}. Call triples $(i,j,k)$ as provided by \Cref{lem:combLem} \define{loop triples (for $x$)}.
Before we give the proof of \Cref{prp:2x}, we introduce the following notation: write~$x_{[i,j]} = x_i\cdots x_j$. For instance, if~$n \geq 4$, then~$x_0 x_1 \cdots x_{n-1} = x_{[0,n-1]} = x_{[0,2]}x_{[3,n-1]}$.

\begin{proof}[Proof of \Cref{prp:2x}]
	Let $x \in \Sigma^n$ be as assumed, and suppose $(q_0,\ldots,q_n)$ is the run of~$\mN$ which accepts~$x$. Observe that if $(i,j,k)$ is a loop triple for~$x$ (by \Cref{lem:combLem} there are at least~$(1-2q)n$ many), then the witnessing NFA $\mN(x)$ has completed at least two loops by the time it has read the word~$\x{0,k-1}$. There are two cases.
	\begin{enumerate}
		\item There exists a loop triple $(\is,j,k)$ for which
		\[
			\max\left(\left|\x{\is,j-1}\right|,\left|\x{j,k-1}\right|\right) > (1-2q)\sqrt{n}.
		\]
		Assume w.l.o.g.\ that $\left|\x{j,k-1}\right| \geq \left|\x{\is,j-1}\right|$ and write
		\[
			x = \x{0,\is-1}\x{\is,j-1}\x{j,k-1}\x{k,n-1}.
		\]
		Since the triple~$(i,j,k)$ is a loop triple, $q_{\is} = q_j = q_k$, and thus~$\mN(x)$ also accepts the word $\x{0,\is-1}\x{j,k-1}\x{\is,j-1}\x{k,n-1}$.
	Since $\mN(x)$ exactly accepts $x$, we have
	\[
		\x{j,k-1}\x{\is,j-1} = \x{\is,j-1}\x{j,k-1}
	\]
	and so \Cref{lynd1} implies $\x{\is,j-1} = z^k$ and $\x{j,k-1} = z^\ell$ for some~$z \in \splus$ and~$k,\ell \in \nn$. Since~$k,\ell \ge 1$, the decomposition of~$\x{\is,k-1}$ trivialises into a product of copies of~$z$:
	\[
		\x{\is,k-1} = \x{\is,j-1}\x{j,k-1} = zz^{k+\ell-1} = z^{k+\ell-1}z
	\]
As $\left|z^{k+\ell-1}\right| \geq \left|\x{j,k-1}\right| \geq (1-2q)n > (1-2q)\sqrt{n}$, we are done.
		\item For all loop triples~$(\is,j,k)$ we have
		\[
			\max\left(\left|\x{\is,j-1}\right|,\left|\x{j,k-1}\right|\right) \leq (1-2q)\sqrt{n}.
		\]
		By \Cref{lem:combLem}, there exist~$(1-2q)n$ indices~$k$ for which there exist~$(\is,j)$ such that~$(\is,j,k)$ is a loop triple. Since every loop in a loop triple has length at most~$ (1-2q)\sqrt{n}$, the pigeonhole principle gives an~$\ell \leq (1-2q)\sqrt{n}$ such that there exist at least~$m \geq \sqrt{n}$ such indices $k$ at which a loop of length~$\ell$ was just completed (hence, we only focus on the \emph{second} loops in each loop triple). Let this set of indices be given in ascending order, denoted by~$\mathcal{K} = \{ k_1,\ldots,k_m \}$, with associated loops~$\rho_1, \ldots, \rho_m \prec x$, each of length~$\ell$.
		
		We show that $\rho_1$ and $\rho_m$ must be disjoint, i.e.\ share no states along their traversals in~$\mN(x)$. Let $q_{k_1}$ be the origin state of the loop $\rho_1$. By definition, $\rho_1$ is the second loop in the loop triple~$(\is_1,j_1,k_1)$.  Suppose $\tau$ is the first loop at $q_{k_1}$ so that $\tau\rho_1$ is a loop triple at $q_{k_1}$. Then, if we read~$b > (1-2q)\sqrt{n}$ letters along the loops at state~$q_{k_1}$, then we could concatenate those loops with $\tau$ to obtain a loop triple, one of whose lengths exceeds~$(1-2q)\sqrt{n}$, which contradicts the assumption of this case. Therefore, at state~$q_{k_1}$, we can only read at most~$(1-2q)\sqrt{n}$ letters of the factors contained in~$\rho_1,\ldots,\rho_m$, before moving on to a different state, never to return. However, by construction, for every~$i \leq m$ we know that~$x_{k_i}$ appears in~$\rho_i$, and thus we must read at least~$m \geq \sqrt{n}$ letters throughout all loops~$\rho_1,\ldots,\rho_m$. Since~$q < 1/2$, we have~$(1-2q)\sqrt{n} < \sqrt{n} \leq m$; hence, the first and last loops $\rho_1$ and $\rho_m$ must be disjoint.
		Thus, $x = u \; \rho_1 \; y \; \rho_m \; u'$ where~$u,y,u' \prec x$ and~$|\rho_1| = |\rho_m| = \ell$.
		By exact acceptance of~$\mN(x)$, we have
		\[
			x = u \; (\rho_1)^2 \; y \; u'
		\]
		since $|\rho_1| = |\rho_m|$. Therefore,~$\rho_1 y = y \rho_m$, and thus, with $y' = \rho_1 y$, we have~$y' \rho_m = \rho_1 y'$.
To show that~$y'$ has the desired length, note that~$y \rho_m$ must contain the set~$\{ x_{k_2},\ldots,x_{k_m} \}$; the loop $\rho_1$, since it is the first loop in $\mathcal{K}$, can only contain~$x_{k_1}$. Since $n \geq 4$, we have
		\[
			|y'| = |y \rho_m| \geq |\mathcal{K}| - 1 = m - 1 \geq \sqrt{n} - 2 \geq \frac{\sqrt{n}}{2}. \qedhere
		\]
	\end{enumerate}
\end{proof}
We now apply \Cref{prp:2x} to go even finer: instead of studying the complexity of~$L_q$, we classify the complexity of \emph{words} in~$L_q$, using plain Kolmogorov complexity. Fix an alphabet~$\Sigma$ of cardinality~$k$, and let~$C_k$ denote the \define{plain Kolmogorov complexity} on words in~$\Sigma$:
\[
	C_k(x) = \min \set{\ell(p)}{U_k(p) = x}.
\]
where~$U_k$ is a universal Turing machine on the $k$-element alphabet.
For details on Kolmogorov complexity, see e.g.~\cite{downey+:alg:2010}.

\begin{proposition}\label[proposition]{prp:Kolm}
	If $x \in \Sigma^n \cap L_q$, then
	\[
		C_k(x) \leq n - \frac{(1-2q)}{2}\sqrt{n} + 5 \log_k(n) + O(1).
	\]
\end{proposition}
Its proof requires an extension of \Cref{lynd1}, which gives a sufficient and necessary criterion for the decomposition of words with same prefix and suffix. As it will be useful to us in the sequel outside of the proof of \Cref{prp:Kolm}, we state it right here in the version of~\cite{Shallit_2008}.

\begin{theorem}[The Second Lyndon-Sch\"utzenberger-Theorem]\label{thm:ls2}
	Let $x,y,z \in \sstar$. Then~$xy = yz$ iff there exist $e \in \nn \setminus \{ 0 \}$, $u \in \splus$ and $v \in \sstar$ such that
	\[
		x = uv, \;\;\;\;\; z = vu, \;\;\;\;\; \text{ and } \;\;\;\;\; y = x^eu = uz^e.
	\]
\end{theorem}
With $|\Sigma| = k$ as before, note that the function which maps $x \in \Sigma^*$ to its $C_k$-witness is an injection. Hence, \Cref{prp:Kolm} immediately yields the following bound on~$|L_q|$.
\begin{corollary}\label[corollary]{usecor27}
	The set $L_q \cap \Sigma^n$ has cardinality in $o\left(|\Sigma|^n\right)$.
\end{corollary}

For the proof of \Cref{prp:Kolm}, we require the following piece of notation. Let~$\ip{\cdot}$ denote the \define{integer part function}; e.g.~$\ip{\frac{3}{2}} = 1$.

\begin{proof}[Proof of \Cref{prp:Kolm}]
	Assume that \Cref{prp:2x} showed there is a word $z \prec x$ which occurs twice, but not as a square\footnote{The case where the square $z^2$ appears is even easier, as less information needs to be coded.}.
In order to code $x$, one only needs to code $z$ as well as the starting positions of its first and second copy inside $x$, plus the remaining bits. The fact that~$|z| \geq (\frac{1-2q}{2})\sqrt{n}$ (which follows from \Cref{prp:2x}) is crucial here.
	Since $z$ appears twice inside~$x$, there exist~$w,w' \prec x$ such that
	\[
		zw = w'z.
	\]We can locate the two copies of $z$ inside $x$ explicitly: define $\ell,\ell',t < n$ such that
	\begin{itemize}
		\item $\ell$ is the starting index of the first copy of $z$ inside $x$;
		\item $\ell'$ is the starting index of the second copy of $z$ inside $x$; and
		\item $t$ is the fist index after the second copy of $z$ inside $x$.
	\end{itemize}
	In particular, $z = \x{\ell,\ell + |z| - 1} = \x{\ell', t-1}$, which we use to write 
	\begin{align*}
		x &= \x{0,\ell-1}zw\x{t,n-1} = \x{0,\ell-1}w'z\x{t,n-1}.\\
		\intertext{For ease of readability, we rewrite this again as}
		x &= x_1zwx_2 = x_1w'zx_2.
	\end{align*}
	We now isolate an upper bound on $C_k(x)$. Let~$m = \ip{\log_k(n)} + 1$, and define the following shorthand: for~$n < k^m - 1$, denote by~$c_n$ the~$k$-ary expression of $n$ in a string of length\footnote{I.e.\ add leading zeroes to fill up the string to length~$m$, if needed. Note that $m$ is \emph{defined} to be sufficiently large for this coding to work.} $m$. Then consider the string
	\[
		c = 0^m1c_{|x_1|}c_{|z|}c_{|w|}c_{|x_2|}x_1wx_2.
	\]
	Since $|z| \geq (\frac{1-2q}{2})\sqrt{n}$, we know that $|x_1wx_2| \leq n - (\frac{1-2q}{2})\sqrt{n}$. Combining this with the fact that
	\[
		|0^m1c_{|x_1|}c_{|z|}c_{|w|}c_{|x_2|}| = 5m + 1
	\]
	we obtain
	\[
		|c| \leq n - \frac{(1-2q)}{2}\sqrt{n} + 5m + 1 \leq n - \frac{(1-2q)}{2}\sqrt{n} + 5\log_k(n) + O(1).
	\]
	One can now compute $x$ from $c$ via the Second Lyndon-Sch\"utzenberger-Theorem.
\end{proof}

From \Cref{usecor27}, we now deduce the Shannon effect for $\cN$. Originally conjectured by Shannon \cite{shannon} and proven (and named) by Lupanov for Boolean functions~\cite{lupanov1,lupanov2}, the Shannon effect says that most strings are of almost maximal complexity. We give a definition due to Wegener \cite{wegener}.
\begin{definition}\label[definition]{dfn:aa}
	Let $P \subset \Sigma^*$. We say that~\define{almost all $x$ have property $P$} if
	\[
		\lim_{n \rightarrow \infty} \frac{|P \cap \Sigma^n|}{|\Sigma|^n} = 1.
	\]
\end{definition}
\begin{definition}
	Let $\Gamma$ be a complexity measure defined on $\Sigma^*$. For~$n \in \nn$, let~$\Gamma\left( \Sigma^n \right) = \max_{x \in \Sigma^n} (\Gamma(x))$. Then~$\Gamma$ has the \define{Shannon effect} if for almost all~$x \in \sstar$ we have
	\[
		\Gamma(x) \geq \Gamma\left(\Sigma^{|x|}\right) - o\left(\Gamma\left(\Sigma^{|x|}\right)\right).
	\]
\end{definition}
By exhibiting upper and lower bounds of complexity for \emph{all} words, it is readily seen that (plain and prefix-free) Kolmogorov complexity satisfy the Shannon effect~\cite{kolmogorov:three:1965,sol1,sol2,levinthesis,levin2}, as do~$A_D$ \cite{shallitWang} and $A_n$ \cite{hydeMA,bjornIncompr}. The cardinality argument of \Cref{usecor27} shows:
\begin{theorem}\label{thm:sh}
	$\cN{}$ satisfies the Shannon effect.
\end{theorem}

\begin{proof}
	Fix $q = 1/(2 + \epsilon)$ for some small $\epsilon > 0$. Since
	\[
		\cN(x) \leq \cNu(x) \leq (|x|/2) + 1
	\]
	by \Cref{lem:poyt}, identifying a suitable lower bound suffices. By \Cref{usecor27}, for $o(|\Sigma|^n)$-many words~$x \in \Sigma^n$ we have~$x \in L_q$. Hence, for almost all (as per \Cref{dfn:aa}) $x \in \Sigma^n$,
	\[
		\frac{n}{2 + \epsilon} \leq \cN(x) \leq \frac{n}{2} + 1
	\]
	and so, for large enough $n$ and $x \in \Sigma^n$, $\cN(x) \in (n/2, n/2 + 1)$, as required.
\end{proof}

\section[Sets of Low-Complexity Words Are Not Context-Free]{$L_q$ Is Not Context-Free}\label{sec:Lq}

Fix $q \in (0,1/2)$ and suppose w.l.o.g.\ that $0,1 \in \Sigma$. In this section, we demonstrate that $L_q$ cannot be generated by a context-free grammar (CFG); hence,~$L_q$ is not context-free. To this end, we first define the concept of a \emph{rich} CFG. We then prove that if a CFG generates~$L_q$, it must be rich. Finally, we show that any rich CFG generates words of arbitrarily high complexity, which contradicts the fact that the CFG generates~$L_q$.

We provide the required definitions. (For more details, see e.g.~\cite{Shallit_2008}.) A \define{context-free grammar} (CFG) is a tuple~$\grammar = (V_T,V_N,S,P)$ where:
\begin{itemize}
	\item $V_T$ is the set of \define{terminal symbols}.
	\item $V_N$ is the set of \define{non-terminal symbols}.
	\item $S \in V_N$ is the \define{start symbol}.
	\item $P$ is a finite set of \define{productions}. 
\end{itemize}
We also insist that~$V_T \cap V_N = \emptyset$, and we define the set of \define{symbols} by~$V = V_T \cup V_N$. Elements of~$V^*$ are called \define{sentential forms}. Productions in~$P$ are pairs~$(A,\gamma)$ where~$A \in V_N$ and~$\alpha \in V^*$. We denote such a production by
\[
	A \rightarrow \gamma.
\]

The \define{derivation relation}~$\Longrightarrow$ is defined as follows: if~$\alpha,\beta \in V^*$ then \define{$\beta$ is derived from~$\alpha$} if~$\alpha = \alpha' A \alpha''$ and~$\beta = \alpha' \gamma \alpha''$ for some~$\alpha',\alpha'' \in V^*$, and there exists a production~$A \rightarrow \gamma$ in~$P$. We write
\[
	\alpha \Longrightarrow \beta.
\]
The transitive and reflexive closure of~$\Longrightarrow$ is denoted by~$\deriv$.

A language that is recognisable by a CFG is called a \define{context-free language}, abbreviated ``CFL''.

\begin{definition}
	A CFG has \define{no useless nonterminals} if:
	\begin{enumerate}
		\item each nonterminal is reachable from the starting symbol; and
		\item a terminal string can be derived from each nonterminal.
	\end{enumerate}
\end{definition}

\begin{definition}
	Let $\grammar$ be a CFG. A nonterminal symbol~$A \in \grammar$ is \define{rich} if for some~$v,w,x,y \in \sstar$ we have both~$vwxy \neq \varepsilon$ and $A \deriv vAx \mid wAy$ as well as:
	\begin{enumerate}
		\item if $vw \neq \varepsilon$ then $vw \neq wv$; and
		\item if $xy \neq \varepsilon$ then $xy \neq yx$.
	\end{enumerate}
	A \define{rich CFG} has a rich nonterminal but no useless nonterminals. A \define{rich CFL} is generated by a rich CFG.
\end{definition}
Our motivation for rich CFGs follows from \Cref{lynd1}, however, we note here that, in style, our richness characterisation is similar to classical results by Ginsburg \cite[Theorem~5.5.1]{ginsburg}, who characterised boundedness of CFLs via syntactical properties of grammars. Our syntactical notion of richness, similarly, characterises the complexity of generated languages, in our case~$L_q$.
The equivalence in \Cref{lynd1} implies that a rich non-terminal can construct words which do not collapse to repeating copies of a common factor $z$. This is needed in \Cref{subsubsub}, where we construct high-complexity words.

\subsection[Only Rich CFGs Can Generate Can Generate Sets of Low-Complexity Words]{Only Rich CFGs Can Generate~$L_q$}

We require the following normal form theorem due to Greibach~\cite{greibach} (see \cite[p.\ 277]{introL} for a modern exposition).

\begin{theorem}[Greibach Normal Form Theorem]\label{thm:gbn}
	Every CFG that has no~$\varepsilon$-productions can be expressed in \define{Greibach Normal Form}: all its production rules are of the form $A \rightarrow x \overline{A}$ where~$x \in \Sigma$ and $\overline{A}$ is a finite word of nonterminals. 
\end{theorem}

The following corollary is immediate.
\begin{corollary}\label[corollary]{gbn}
	Every CFL omitting $\varepsilon$ is generated by a CFG in Greibach Normal~Form.
\end{corollary}

Our main result in this subsection is the following.

\begin{theorem}\label{thm:gammaRich}
	If $L_q$ is generated by a context-free grammar~$\grammar$, then $\grammar$ is rich.
\end{theorem}

\begin{proof}
	By our results in the previous section,~$L_q$ is non-empty; further, by definition,~$\varepsilon \not\in L_q$. So, by \Cref{gbn}, there exists a CFG~$\grammar$ in Greibach Normal Form which generates $L_q$. We show that $\grammar$ must be rich by a counting argument on the number of nonterminals of~$\grammar$.
	Let~$k \in \nn$ denote the number of nonterminals in~$\grammar$. Define
	\begin{align*}
		&x_i= 0^i 1^{4k-i}  &\text{for $i = 1,2,\ldots,4k-1$.}
	\end{align*}
	By \Cref{prp:longgg}, for every $i$ there exists $m_i \in \nn$ for which $x_i^{m_i} \in L_q$. Similarly, for each~$i$ there exists~$m'_i \in \nn$ for which the derivation tree of~$x_i^{m'_i}$ has a branch which contains some nonterminal~$A$ at least~$(4k)^2+1$ times. Let~$M \in \nn$ be sufficiently large to satisfy these requirements for all~$x_i$ simultaneously. By the pigeonhole principle, there exist $i,j,\ell \leq 4k-1$ such that some nonterminal $A$ appears at least~$(4k)^2 +1$ times in some branch of the derivation tree of each of~$x_i^M, x_j^M$ and~$x_\ell^M$.
	
	Consider such a sufficiently long branch of the derivation tree of $x_i^M$, in which we choose to expand~$A$ at the end. Since~$\grammar$ is in Greibach Normal Form, the derivation is of the form
	\[
		S \deriv y_i^1 y_i^2 y_i^3 \ldots y_i^s A z_i^s \ldots z_i^3 z_i^2 z_i^1
	\]
	from which $x_i^M$ can be derived in at least $(4k)^2 + 1$ expansions of $A$. Observe that each~$y_i^j \neq \varepsilon$, since $\grammar$ is in Greibach Normal Form. Consider the number of expansions of~$A$ in terms of blocks~$B_1,B_2,\ldots,B_n$ such that each block has cardinality $4k$. By assumption, $n \geq 4k + 1$. Let~$A_m$ be the derivation of $A$ from the expansions in block $B_m$. There are two cases:
	\begin{enumerate}
		\item For some $m \leq n$, $A_m = yA$ with $y \in \sstar$ and\footnote{This is a consequence of the observation immediately following \Cref{thm:gbn}.} $|y| \geq 4k$.
		\item For all $m \leq n$, $A_m = y_m A z_m$ with $y_m, z_m \in \splus$. Then, $A_{4k} = y A z$ where we have~$|y|, |z| \geq 4k$.
	\end{enumerate}
Since these two cases apply to all $x_i^M, x_j^M$ and $x_{\ell}^M$, two of them must share the same case above. W.l.o.g.\ assume both $x_i^M$ and $x_j^M$ fall into case 2 (the argument for case 1 is similar). Hence, $T \deriv y A z$ (from the derivation of $x_i^M$) and $T \deriv v A w$ (from the derivation of $x_j^M$) where~$|y|,|z|,|v|,|w| \geq 4k$. By definition,~$y,z$ contain~$i$ zeroes, while~$v,w$ contain~$j$ zeroes among the first~$4k$ letters. It is now seen from the First Lyndon-Sch\"utzenberger-Theorem that~$yv \neq vy$ and~$zw \neq wz$. Hence, $A$ is a rich nonterminal.
\end{proof}

\subsection{Every Rich CFG Generates High-Complexity Words}\label{subsubsub}

In this section, we prove that every rich CFG generates words of arbitrarily high complexity relative to its length. In particular, there exists a word $x$ for which $\cN(x) > q|x|$ for every~$q \in (0,1/2)$. This contradicts the fact that any rich CFG can generate $L_q$ for any~$q$, since any~$x \in L_q$ satisfies $\cN(x) < q|x|$. We also isolate the following technical proposition.

\begin{proposition}\label[proposition]{prp:isol}
	Suppose $u,v \in \Sigma^n$ with $uv \neq vu$. Then the following set is infinite:
	\begin{align*}
		\ii_{(u,v)} = \{ \, x \in \{ u,v \}^* : &\text{if $y \prec x$ satisfies $|y| > 2\log(|x|)$}\\ 
		&\qquad\qquad \text{then $y$ occurs exactly once in $x$} \}
	\end{align*}
\end{proposition}
Proving \Cref{prp:isol} takes a few technical lemmas on the behaviour of non-commuting strings in formal languages. Denote the \define{set of factors} of~$w \in \sstar$ by~$[w] = \set{x \in \sstar}{x \prec w}$. For convenience, we now fix some $n \in \nn$ and a pair~$u,v \in \Sigma^n$ for which~$uv \neq vu$.

\begin{lemma}\label[lemma]{lem1}
	$uv, vu \not\in [u^3] \cup [v^3]$
\end{lemma}

\begin{proof}
	We give the argument for $uv \not\in [u^3]$; the other parts are similar. Assume that $uv \in [u^3]$; thus write $u^3 = xuvy$ for some $x,y \in \sstar$. Note that $|xy| = |u| = |v|$. Since~$uv \neq vu$ we cannot have $x,y \in \{ u,v \}$, and thus $|x|, |y| < |u| = |v|$. But now, by periodicity of~$u^3$, we must have $xy = u$. Thus $u^3 = xyxyxy = xuvy$. Therefore, we have~$uv = yxyx$, from which it follows that~$u = xy = yx = v$, contradicting~$uv \neq vu$.
\end{proof}
To motivate the next lemma, we introduce string homomorphisms.

\begin{definition}
	A function $h \colon \{ 0,1, \ldots,n-1 \}^* \rightarrow \sstar$ is a \define{string homomorphism} if for all~$n_i \in \{ 0,1,\ldots,n-1 \}$ we have $h(n_0\cdots n_k) = h(n_0)\cdots h(n_k)$.
\end{definition}
Observe that every such string homomorphism is uniquely defined by its action on the alphabet. Define a string homomorphism $h \colon \{ 0,1,2 \} \rightarrow \{ u,v \}^*$ given by
\begin{align*}
	h(0) &= uv & h(1) &= vu &  h(2) &= u^3 v^4
\end{align*}
With this string homomorphism fixed, the following lemma is immediate from \Cref{lem1}.

\begin{lemma}\label[lemma]{lem2}
	$u^4,v^4 \not\in \bigcup \set{[x]}{x \in h(\{ 0,1 \}^*)}$ 	
\end{lemma}
To give a proof of \Cref{prp:isol}, we first code words as follows. For every~$k \in \nn$, let~$\sigma_k$ be the lexicographical concatenation of all positive integers which, coded in binary, have length~$k$; each is then followed by a $2$. For instance, $\sigma_2 = 002012102112$.
We consider the images of these words under $h$, and collect some immediate properties of the~$\sigma_k$ and the~$h(\sigma_k)$ below, whose proofs are readily deduced, hence omitted.

\begin{lemma}\label[lemma]{lemT}
	Let $k \in \nn$.
	\begin{enumerate}
		\item $|\sigma_k| = 2^k(k+1)$ \label{item1}
		\item $|h(\sigma_k)| = 2^k|v|(2k + 7)$ \label{item2}
		\item $2 \log(|h(\sigma_k)|) \geq 2k + 14|v|$ for sufficiently large $k$. \label{item3}
	\end{enumerate}
\end{lemma}
To prove \Cref{prp:isol}, we show that for large enough~$k$, every substring of~$h(\sigma_k)$ of length at least~$2\log|h(\sigma_k)|$ must contain two copies of~$h(2)$; since the word between any two copies of~$h(2)$ is unique within~$h(\sigma_k)$, the proposition is proven.

\begin{proof}[Proof of \Cref{prp:isol}]
	Fix $k \in \nn$ sufficiently large so as to satisfy \cref{item3}, and consider the word $h(\sigma_k)$. By construction and the choice of $k$, if $y \prec h(\sigma_k)$ and
	\[
		|y| \geq 2\log(|h(\sigma_k)|)
	\]
	then~$y$ contains two copies of~$h(2)$. By definition, $v^4 \prec h(2)$; on the other hand, \Cref{lem2} shows that~$v^4$ cannot be a factor of any~$h(w)$ with~$w \in \{ 0,1 \}^*$. Hence,~$v^4 \prec y$ must be a factor of some~$h(2)$ occurring in~$h(\sigma_k)$. We show that the copies of~$h(2)$ in~$y$ and in~$h(\sigma_k)$ overlap perfectly. Consider the word~$h(2)z = xh(2)$ contained in $h(\sigma_k)$. This bordered word is in fact a proper square, which can seen by a case analysis on $x$. Write
	\begin{align*}
		x &= a|u| + \ell &\text{for some $a \in \nn, 0 \leq \ell < |u|$}\\
		u &= \alpha\beta &\text{where $|\alpha| = \ell$}\\
		v &= \gamma\delta &\text{where $|\gamma| = \ell$}
	\end{align*}
	and note that this renaming implies $|\beta| = |\delta|$, and that
	\[
		h(2)z = xh(2) = x_1x_2(\alpha\beta)^3(\gamma\delta)^4 = (\alpha\beta)^3(\gamma\delta)^4z.
	\]
	There are four cases describing $x_1$; using Theorems \ref{lynd1} and \ref{thm:ls2}, each will lead to a contradiction.
\begin{enumerate}
		\item $a = 0$:
		Since $|\alpha| = \ell$, this case implies $\alpha\beta = \beta\alpha$. In addition, we also have~$|\beta\gamma| = |\alpha\beta| = |\beta\alpha|$, and so~$\alpha = \gamma$. By comparing lengths, it is easily seen that~$\beta = \delta$, and so $u = \alpha\beta = \gamma\delta = v$, a contradiction.
		\item $a = 1$:
		Since $u = \alpha\beta$, by comparing initial segments it is readily seen that in this case~$uv \prec v^4$, contradicting \Cref{lem1}. 
		\item $a = 2$ or $a = 3$:
		Since $xh(2) = h(2)z$, we have~$|x| = |z|$. So, if $a = 2,3$, then again by comparing initial segments it is readily seen that $uv \prec v^4$, contradicting \Cref{lem1}.
		\item $a > 3$:
		Here,~$v^4 \prec z$. Since $z = h(w)$ for some $w \in \{ 0,1 \}^*$, \Cref{lem1} gives a contradiction.
\end{enumerate}
Hence, the copies of $h(2)$ appearing in $y$ are exactly those appearing in $h(\sigma_k)$. But now, if~$y' \prec h(\sigma_k)$ is of length at least $2\log(|h(\sigma_k|)$, then it contains a factor of the form~$h(2)\rho h(2)$ where $\rho \in h(\{ u,v \}^*)$. Each such $\rho$ appears only once in $h(\sigma_k)$, by construction. Thus, for large enough $k$, the word $h(\sigma_k)$ is as required, and thus the set $\ii_{(u,v)}$ is infinite.
\end{proof}

\begin{theorem}\label{thm:richHigh}
	If $\grammar$ is a rich CFG, then $\grammar$ generates a word~$x \in \sstar$ such that~$\cN(x) > q|x|$ for every~$q \in (0,1/2)$.
\end{theorem}
For notation, if $\sigma \in \{ 0,1 \}^*$, let $\overline{\sigma}$ denote the reverse of $\sigma$. Further, if $x,y \in \sstar$ satisfy~$xz = zy$ and both $xz,zy \prec w$ such that $xz$ and $zy$ overlap at $z$, then call~$xzy$ its \define{union}, written as~$xz \cup zy$. We use \Cref{prp:isol}.

\begin{proof}
	Let $\grammar$ be a rich CFG with rich nonterminal~$A$ and witnesses~$x,y,x',y' \in \sstar$ for which~$A \deriv xAy \mid x' A y'$ and~$xx' \neq x'x$ and $yy' \neq y'y$. Hence, define the words~$u_1 = xx', v_1 = x'x$ and~$u_2 = yy', v_2 = y'y$. Now define string homomorphisms~$g,h$ by:
	\begin{align*}
		g(0) &= u_1v_1 & g(1) &= v_1u_1 & g(2) &= u_1^3v_1^4\\
		h(0) &= u_2v_2 & h(1) &= v_2u_2 & h(2) &= u_2^3v_2^4
	\end{align*}
	Fix any $w_1,w_2,w_3 \in \sstar$ for which
	\begin{align}
		S \deriv w_1Aw_3 \;\;\;\;\;\; \text{ and } \;\;\;\;\;\; A \deriv w_2. \tag{$*$}
	\end{align}
	By repeated application of the generation rules in $(*)$, it is readily seen that for any~$m,k \in \nn$, the word $y_{m,k}$ of the following form is generated by $\grammar$:
	\[
		y_{m,k} = (w_1 \; g(\sigma_k)) \; (x^m \; w_2 \; y^m) \; (\overline{h(\sigma_k)} \; w_3).
	\]
We show that, for sufficiently large $m,k$, the word $y_{m,k}$ has large non-deterministic automatic complexity. Choose $m,k$ large enough so that $|y_{m,k}| \gg |w_1w_2w_3|$, and let~$n = |y_{m,k}|$. Since we may choose~$k,m$ freely, we may also impose that
	\begin{align}
		2 \log(n) \leq m|x| \leq 3 \log(n) \in o(\sqrt{n}). \tag{$\dagger$}
	\end{align}
Now, let~$z \prec y_{m,k}$ whose length is in~$O(\sqrt{n})$ be the first occurrence of a word in~$y_{m,k}$ of the form~$zb = cz$ for words~$b,c \prec y_{m,k}$. Below, we show that this is only possible if $b = c = \varepsilon$.
	
By choosing $m,k$ wisely, we may assume that~$|z|$ is even. Further, it will be convenient to distinguish the words which make up the left-hand and right-hand squares of~$z$; hence, write~$z = z_1z_2 = z_1'z_2'$ so that~$|z_1| = |z_2|$ and~$z_1 = z_1', z_2 = z_2'$, and $z_1z_2b = cz_1'z_2'$.

	We show that~$z_1 \prec w_1 g(\sigma_k)$; the case that~$z_2' \prec \overline{h(\sigma_k)} w_3$ is similar. Note that, otherwise, we may choose~$k$ large enough so that~$z_1$ intersects~$\overline{h(\sigma_k)} w_3$, and in particular, we may enforce that this intersection $s \in \sstar$ has length at least~$2\log(n)$. By construction and the fact that~$z_1z_2b = cz_1'z_2'$, the word~$s$ must appear twice in~$\overline{h(\sigma_k)}$, which contradicts \Cref{prp:isol}.\footnote{See in particular the proof of \Cref{prp:isol} to note that~$\ii_{(u_2,v_2)}$ can be generated by sets of the form~$h(\sigma_k)$, and by a similar argument, by those of the form~$\overline{h(\sigma_k)}$.}
	
	Thus, $z_1 \prec w_1g(\sigma_k)$ and $z_2' \prec \overline{h(\sigma_k)}w_3$ imply that~$x^m w_2 y^m$ is a factor of the union~$z_2b \cup cz_1'$. By a counting argument, it is seen that either $x^m \prec z_2$ or $y^m \prec z_1'$. If~$x^m \prec z_2$---the other case is similar---then also~$x^m \prec z_2' \prec \overline{h(\sigma_k)}$. But this is impossible by \Cref{prp:isol} and since~$|z_2'| \geq m|x| \geq 2\log(n)$ by $(\dagger)$. Therefore, we have arrived at a contradiction: we can only have $z_1z_2b = cz_1'z_2'$ if $b = c = \varepsilon$.
But, the contrapositive of \Cref{prp:2x} shows that~$y_{m,k} \not\in L_q$ for every~$q \in (0,1/2)$. Since $y_{m,k}$ is generated by $\grammar$, the result is proven.
\end{proof}
\Cref{thm:richHigh} and \Cref{thm:gammaRich} combined imply our main result of this section:

\begin{theorem}\label{thm:Lq}
	For every $q \in (0,1/2)$, the language $L_q$ is not context-free.
\end{theorem}
A language~$L$ is \define{CFL-immune} if it contains no infinite context-free language as a subset. We note here that $L_q$ cannot be CFL-immune, since for every~$x \in \Sigma$, the regular language~$\{ x \}^+$ is contained in $L_q$ (modulo finitely many words, depending on $q$), and each of its words has constant complexity. However, the following holds:

\begin{theorem}\label{thm:compl}
	For every $q \in (0,1/2)$, the language $\sstar \setminus L_q$ is CFL-immune.
\end{theorem}

\begin{proof}
	Recall that $\sstar \setminus L_q = \set{x \in \sstar}{\cN(x) \geq q|x|}$. By the Pumping Lemma for CFGs, if~$L$ is an infinite context-free language, then it contains a set $L'$ of the form
	\[
		L' = \set{ua^\ell v b^\ell w}{u,v,w \in \sstar \land a,b \in \splus \land \ell \geq 0}.
	\]
	We show that $\sstar \setminus L_q$ cannot contain any such $L'$, hence,~$\sstar \setminus L_q$ cannot contain an infinite CFL. Consider some such $L'$ and denote its defining word by
	\[
		\alpha(\ell) = ua^\ell v b^\ell w.
	\]
	We show that~$\cN(\alpha(\ell)) < q|\alpha(\ell)|$ for large enough $\ell$, proving that~$L' \cap (\sstar \setminus L_q)$ is finite.
	
	Consider $\alpha(\ell)$ with base words $a,b$. With~$k = \ip{\frac{3}{q}} + 1$, define the repetition number~$\ell'$ by
	\[
		\ell' = (mk|a||b|) + |b|k.
	\]
Note that $\ell'$ depends on $m$. Now, letting $i_0 = k|a|$ and $j_0 = k|b|$, rewrite $\alpha(\ell')$ as
	\begin{align*}
		\alpha(\ell') = u \; a^{\ell'} \; v \; b^{\ell'} \; w = u \; \left(a^{m|b|} \right)^{i_0} \; a^{k|b|} \; v \; \left(b^{m|a| + 1} \right)^{j_0} \; w.
	\end{align*}
We claim that, for large enough $m$, there exists exactly one accepting run in the automaton~$\mN(\alpha(\ell'))$; the one in which the loop $a^{m|b|}$ is taken exactly $i_0$ times, and, similarly, $b^{m|a| + 1}$ is taken $j_0$ times. To see this, suppose there exists a pair~$(i,j)$ for which~$(i_0+i,j_0-j)$ is a pair of positive naturals, and
\[
	\left| \left( a^{m|b|} \right)^{(i_0 + i)} \; \left( b^{m|a| + 1} \right)^{(j_0-j)} \right| = \left| \left( a^{m|b|} \right)^i \; \left( b^{m|a| + 1} \right)^j \right|
\]
which readily reduces to the Diophantine equation
\[
	i(m|a|) + j(-(m|a| + 1)) = 0.
\]
A particular solution is $(i,j) = (m|a| + 1, m|a|)$, and, hence, the set of general solutions is given by the following (cf.\ for instance \cite[p.\ 34]{stillwell} for a proof of this classical fact):
\[
	S = \set{((m|a| + 1)(1-t), (m|a|)(1-t))}{t \in \zz} = \set{((m|a| + 1)t, m|a|t)}{t \in \zz}
\]
Note that the solution $t=0$ corresponds to our choice of~$(i_0,j_0)$.
We show that, once~$m$ is large enough, no other solution for~$t$ is possible. To see this, note that e.g.\ $t = 1$ implies~$i = m|a| + 1$ and~$j = m|a|$. However, for large enough~$m$, we then have~$j_0 - j = k|b| - m|a| < 0$, which does not make sense---one cannot traverse a loop a negative number of times. This proves exactness. Now, note that for sufficiently large~$m$ we have
\begin{align*}
	\cN(\alpha(\ell')) &\leq |u| + |a|m|b| + |a|k|b| + |v| + |b|(m|a| + 1) + |w|\\
	&= 2m|a||b| + \mathsf{const}
\intertext{while}
	|\alpha(\ell')| &= |u| + \ell'|a| + |v| + \ell'|b| + |w|\\
	&= \ell'(|a| + |b|) + \mathsf{const}\\
	&= (mk|a||b|)(|a| + |b|) + \mathsf{const}.
\intertext{We now complete the proof by noting that}
	\cN(\alpha(\ell')) &\leq 2m|a||b| \leq q \left( m \left(\frac{3}{q} \right) |a||b| \right)(|a| + |b|)\\
	&< q \left( mk|a||b| \right)(|a| + |b|)\\
	&\leq q|\alpha(\ell')|. \qedhere
\end{align*}
\end{proof}

Since~$\cN(x) \leq A_D(x)$ for all words~$x$, \Cref{thm:compl} also implies: 

\begin{corollary}\label[corollary]{cor:ADextension}
For every $q \in (0,1/2)$, $\set{x \in \sstar}{A_D(x) \geq q|x|}$ is CFL-immune.
\end{corollary}

We conjecture at this stage that recognising $L_q$ in linear space is equivalent to recognising it via a linearly bounded automaton. Using the fact that
\[
	\cN(x) \le \cNu(x) \le 1 + (|x|/2)
\]
and via an encoding of NFAs in the style of~\cite{shallitWang}, this should yield a linear-space brute-force algorithm. However, a careful check is required to claim this as a theorem.

\section[Sets of Low-Complexity Words Cannot Be Recognised by Certain Constant-Depth Circuits]{$L_q$ Cannot Be Recognised by Certain Constant-Depth Circuits}\label{sec:circ}

In this section, we expand on our work in \Cref{sec:Lq} by investigating the complexity of $L_q$ further. Instead of considering pushdown automata, in this section we consider constant-depth circuits. We show that two types of circuits cannot recognise $L_q$ either, which is analogous to \Cref{thm:Lq} for pushdown automata.

Fix $q \in (0,1/2)$ and fix $\Sigma = \{ 0,1 \}$. We first introduce two types of constant depth circuits explicitly---the class $\sac$ in \Cref{ss:c1}, and $\psac$ in \Cref{ss:c2}---and then show that neither can recognise $L_q$, nor its complement.

\subsection[The Circuit Class \textbf{SAC}0]{The Circuit Class $\sac$}\label{ss:c1}

Suppose $k \geq 1$. A language $L$ is called \define{$\mathbf{SAC}^k$-recognisable} if it is recognised by a polynomial-size,~$O(\log^k n)$-depth, uniform\footnote{Requiring uniformity is debatable; see e.g.\ \cite[Remark 29]{bjornMax}.} semi-unbounded fan-in circuit (a circuit is \define{semi-unbounded} if it has unbounded fan-in-$\lor$, bounded fan-in-$\land$, and admits negative literals but no other negations~\cite{borodin}; cf.\ \cite{venk,bjornMax}). Of these classes of particular interest is $\mathbf{SAC}^1$, since it is the class $\mathbf{logCFL}$ of languages which are log-space reducible to context-free languages \cite{sud,venk}. More generally,~$\mathbf{SAC}^k$ enjoys the following relationship with the classical classes~$\mathbf{AC}^k$ and~$\mathbf{NC}^k$:
\begin{align*}
	&\mathbf{NC}^k \subseteq \mathbf{SAC}^k \subseteq \mathbf{AC}^k \subseteq \mathbf{NC}^{k+1} &\text{for all $k \geq 1$.}
\end{align*}
Just like $\mathbf{NC}^k$ and $\mathbf{AC}^k$, the class $\mathbf{SAC}^k$ is also closed under complements \cite[Corollary~15]{borodin}.

Here, we consider the class~$\sac$. Contrary to the classes above,~$\sac$ is \emph{not} closed under complementation \cite{borodin}. Note that $\sac$-circuits have \emph{constant} depth; hence, the $\sac$-recognisable languages can be characterised by formulas in a simple propositional language, as expressed in \Cref{lemt}. We give a formal definition of $\sac$ due to Kjos-Hanssen~\cite{bjornMax}.

\begin{definition}\label[definition]{dfn:cccc}
	A language $L \subset \{ 0,1 \}^*$ is \define{$\sac$-recognisable} if there exists a family~$(C_i)_{i < \omega}$ of Boolean circuits which recognises $L$ and which satisfies the following:
	\begin{enumerate}
		\item Each $C_i$ is defined over the basic set $\{ \land,\lor \}$ and accepts negative literals.
		\item The family $(C_i)_{i < \omega}$ has constant depth.
		\item Each $C_i$ has unbounded fan-in-$\lor$ and bounded fan-in-$\land$.
		\item Each $C_i$ accepts words of length $i$.
	\end{enumerate}
\end{definition}

Note that, for the classes $\mathbf{SAC}^k$ with~$k > 0$, the size of the circuit must be polynomial in~$n$. However, this requirement is redundant for $\sac$ (cf.~\cite[Remark~30]{bjornMax}).
An important characterisation of $\sac$-recognisable languages follows (cf.\ \cite{bjornMax} and \cite[p.~560]{borodin}).

\begin{lemma}\label[lemma]{lemt}
	A language $L \subset \sstar$ is $\sac$-recognisable iff there exists $c \in \nn$ such that: for every~$n \in \nn$ and every $x \in \Sigma^n$ there exist~$k_n \in \nn$ and a Boolean formula~$\psi_n = \bigvee^{k_n}_{i=1} \varphi_{i,n}$ for which~$\varphi_{i,n}$ is a conjunction of at most $c$ literals, and
	\[
		x \in L \iff \psi_n(x) \text{ holds.}
	\]
\end{lemma}
Using this lemma, which can be deduced from the distributive properties of propositional languages, we conclude:

\begin{theorem}\label{thm:sac1}
	For $q \in (0,1/2)$ and $|\Sigma| = 2$,~$L_q \not\in \sac$ and $\sstar \setminus L_q \not\in \sac$.
\end{theorem}

\begin{proof}
The proof uses a counting argument using \Cref{lemt}. First, suppose that~$L_q \in \sac$, witnessed by a sequence of formulas $(\psi_n)_{n < \omega}$ and a constant~$c$. Consider~$\psi_n$. Since~$\varphi_{1,n}$ mentions at most~$c$ variables, the circuit accepts every word which agrees on these~$c$ variables. Hence,~$\psi_n$ accepts at least~$2^{n-c}$ words. Yet the order of $L_q$ is in~$o\left(2^n\right)$, by \Cref{usecor27}, which contradicts the fact that~$(\psi_n)_{n < \omega}$ recognises~$L_q$.
	
Now, suppose $\sstar \setminus L_q \in \sac$, again accepted by~$(\psi_n)_{n < \omega}$ with constant~$c$. Separate the positive from the negative literals in $\varphi_{1,n}$; there are at most $c' \leq c$ such positive literals. Thus, for any word~$x = x_1\cdots x_n \in \sstar$, if~$x_i = 1$ for all such positive literals, and~$x_i = 0$ everywhere else, then~$\psi_n$ accepts~$x$. But for large enough~$n$, such~$x$ is in~$L_q$ by \Cref{prp:gapsPrp}, which contradicts the fact that~$(\psi_n)_{n < \omega}$ recognises~$\sstar \setminus L_q$. 
\end{proof}

\subsection[The Circuit Class +\textbf{SAC}0]{The Circuit Class $\psac$}\label{ss:c2}

In this section, we consider the class $\psac$, whose definition differs from that of $\sac$ only in the choice of base set. Let $\oplus$ denote the $\mathsf{XOR}$ operation. As before, fix~$q \in (0,1/2)$.

\begin{definition}\label[definition]{psacdfn}
	A language $L \subset \{ 0,1 \}^*$ is \define{$\psac$-recognisable} if there exists a family~$(C_i)_{i < \omega}$ of Boolean circuits which recognises $L$ and which satisfies the following:
	\begin{enumerate}
		\item Each $C_i$ is defined over the basic set $\{ \land,\xor \}$ and accepts negative literals.
		\item The family $(C_i)_{i < \omega}$ has constant depth.
		\item Each $C_i$ has unbounded fan-in-$\xor$ and bounded fan-in-$\land$.
		\item Each $C_i$ accepts words of length $i$.
	\end{enumerate}
\end{definition}
From this definition and the following observation, we can investigate languages larger than binary. Recall that in the previous subsection, we focussed solely on the two-element alphabet~$\{ 0,1 \}$. This was forced by the fact that Boolean expressions have trouble expressing Boolean operations on non-binary languages (e.g.\ what does the formula~$0 \land 2$ evaluate to?). This can be remedied in the class $\psac$ for some languages, courtesy of the operator~$\xor$.

It is readily seen that $(\{ 0,1 \}, \xor, \land)$ is isomorphic to the field of two elements
\[
	\ff_2 = (\zz/2\zz ,  + , \times).
\]
(Studying Boolean circuits in terms of the arithmetic of~$\ff_2$ goes back to G\'al and Wigderson \cite{gal}. We also mention here similarities to the work of Razborov-Smolensky \cite{razb,smol,smol2}.) To extend this equivalence beyond binary alphabets, take the field~$\ff_{p}$ for some prime~$p > 2$. By interpreting~$(\xor,\land)$ as $(+,\times)$ mod~$p$, we extend~$\sac$-recognisability to alphabets of prime cardinality. Below, we give a natural characterisation of~$\psac$-recognisability in terms of propositional formulas, similar to that in \Cref{lemt}.\footnote{For a classical definition of $\psac$ in terms of the complexity of Boolean circuits see e.g.~\cite[Section~4]{bjornMax}.}

\begin{definition}\label[definition]{lemtt}
	Let $|\Sigma| = p$ for some~$p \in \pp$. Then~$L$ is~$\psac$-recognisable if there exists~$c \in \nn$ such that: for every~$n \in \nn$ and every $x \in \Sigma^n$ there exists $k_n \in \nn$ and a formula~$\psi_n = \bigoplus^{k_n}_{i=1} \varphi_{i,n}$ for which~$\varphi_{i,n}$ is a conjunction of at most $c$ literals and
	\[
		x \in L \iff \psi_n(x) \neq 0.
	\]
\end{definition}

\begin{remark}
Observe that there is a subtle difference between $\sac$ and $\psac$ in the case~$p=2$. An $\sac$ circuit accepts a word $x \in \Sigma^n$ if \emph{any} term in the disjunction of $\psi_n(x)$ holds. On the contrary, in $\psac$, the disjunction is interpreted as addition modulo $2$, and, hence,~$x$ is accepted only if the number of terms in the disjunction of $\psi_n$ is odd.
Also, note that \Cref{lemtt} requires a real-world formalism in which gates are able to carry out addition and multiplication modulo~$p$ as a primitive. This assumption is not needed when~$p=2$, as such Boolean circuits can be modelled using $\xor$ and $\land$, as mentioned.
\end{remark}
For completeness, we mention here that $\sac \neq \cosac$ (see \cite{borodin}), while we have~$\copsac = \psac$ (inverting a polynomial in a finite field requires only a constant number of layers; we use this fact in the proof of \Cref{thm:sac2}). Further, we know~$\sac \not\subseteq \psac$ \cite[Theorem 39]{bjornMax}.

Below, we prove the following complexity characterisation of alphabets of prime cardinality.

\begin{thm:ssss}
	Let $q \in (0,1/2)$ and $|\Sigma| = p$ for some prime $p$. Then $L_q \not\in \psac$ and~$\sstar \setminus L_q \not\in \psac$.
\end{thm:ssss}

By translating prime-cardinality-alphabets into finite fields, we may use the tools of field theory. In this section, we collect facts about finite fields which we require to prove \Cref{thm:sac2}.

\begin{lemma}\label[lemma]{propsFF}
	Let $\ff$ be a finite field.
	\begin{enumerate}
		\item By prime decomposition, $\ff$ has prime characteristic.
		\item $\ff$ has order $p^n$ for some $p \in \pp$. \cite[33.2, 33.10]{fraleigh}
		\item If $\ff$ has order~$p^n$ then~$\ff$ has characteristic~$p$. \cite[Sec.\ 14.3]{dummit}
		\item For every $p \in \pp$ and $n \in \nn$, there is one field up to isomorphism of order~$p^n$~\cite[33.12]{fraleigh}. This field has a subfield of order $p$, the \emph{prime subfield}.
		\item All functions from $\ff$ to itself are polynomials. \cite[Exercises 22: 31.c.] {fraleigh} \label{poly}
		\item If $\ff$ has order $p^n$ and $x \in \ff$ then $x^{p^m} = x^{p^{m+n}}$ for all $m \in \nn$. In particular, $x = x^{p^n}$, since the multiplicative subgroup of $\ff$ has order $p^n - 1$. \cite[p.\ 550]{dummit} \label{expo}
	\end{enumerate}
\end{lemma}
If $p \in \pp$ and $n \in \nn$, let $\ff_{p^n}$ denote the (unique up to isomorphism) field of order~$p^n$.

\begin{lemma}\label[lemma]{trace}
Suppose $\varphi \colon \ff_{p^n} \rightarrow \ff_{p}$ is linear, i.e.
\[
	\varphi(x+y) = \varphi(x) + \varphi(y) \;\;\;\;\; \text{ and } \;\;\;\;\; \varphi(ax) = a\varphi(x)
\]
for all $x,y \in \ff_{p^n}$ and $a \in \ff_p$. Then there exist $a_1,\ldots,a_n \in \ff_{p^n}$ for which
\[
	\varphi(x) = \sum_{i=1}^n a_i x^{p^i}.
\]
In fact, every linear function from $\ff_{p^n}$ to $\ff_p$ arises in this way.
\end{lemma}
For a proof and related details on \emph{field traces}, see for instance~\cite[Theorem~2.24]{lidl} and~\cite[Chapter~2.3]{lidl}.
In fact, their proof shows that there exists \emph{one}~$z \in \ff_{p^n}$ for which $a_i = z^{p^i}$.
We now give a characterisation of $\psac$ in terms of finite fields and their operations. This characterisation is akin to that of $\sac$ in \Cref{lemt} in terms of propositional formulas.

\begin{proposition}\label[proposition]{isoProp}
	Let $\phi_n \colon \ff_p^n \rightarrow \ff_{p^n}$ be a linear isomorphism of vector spaces over~$\ff_p$, and suppose~$L \subset \sstar$ is~$\psac$-recognisable. Then there exists a family of polynomials $(f_n)_{n \in \nn}$ with~$f_n \colon \ff_{p^n} \rightarrow \ff_p$ for which
	\[
		x \in L \cap \Sigma^n \iff (f_n \circ \phi_n)(x) \neq 0
	\]
	and for which there exists $\ell \in \nn$ such that for all $n \in \nn$ we have $\deg(f_n) \leq p^n - p^{n-\ell}$.
\end{proposition}

\begin{proof}
As we work in $\ff_{p}$, we identify $
\oplus$ with addition modulo $p$, and write~$x + y$ for~$x \oplus y$. Consider the family~$(\psi_n)_{n \in \nn}$ given by \Cref{lemtt}. Therefore, there exists~$k_n \in \nn$ for which
	\[
		\psi_n(x) = \sum_{i = 1}^{k_n} \left(  \prod_{j = 1}^{m_i} \pi_{(i,j)}(x) \right)
	\]
	where~$\pi_{(\cdot,\cdot)}$ is a projection function from~$\ff_p^n$ to~$\ff_p$. Note that since the Boolean circuit has constant depth, the sequence~$(m_i)_{i \in \nn}$ is bounded. Consider the composition~$f_n$ defined by:
	\begin{align}\label{pppp}
		f_n(x) = (\psi_n \circ \phi_n^{-1})(x) = \sum_{i=1}^{k_n} \left( \prod_{j=1}^{m_i} \left(\pi_{(i,j)} \circ \phi_n^{-1}\right)(x) \right)
	\end{align}
	Note that $f_n \circ \phi_n = \psi_n$ and thus $x \in L \cap \Sigma^n$ if and only if $\psi_n(x) = (f_n \circ \phi_n)(x) \neq 0$; so, $f_n$ is as needed. We now show that $f_n$ is a polynomial. Since $\pi_{(\cdot,\cdot)}$ and $\phi_n^{-1}$ are linear, so is their composition, whose range is contained in $\ff_p$. \Cref{trace} tells us now that $\pi_{(i,j)} \circ \phi_n^{-1}$ may be expressed as
	\[
		\left( \pi_{(i,j)} \circ \phi_n^{-1} \right)(x) = \sum_{t = 1}^{n} a_{(i,j,\ell)} x^{p^t}.
	\]
	Therefore, eq.~(\ref{pppp}) shows that~$f_n$ itself is a polynomial on~$\ff_{p^n}$ with range in~$\ff_p$. To bound the degree of~$f_n$, use distributivity in the field~$\ff_p$ and \Cref{trace} to write
	\begin{align*}
		f_n(x) = \left(\psi_n \circ \phi_n^{-1}\right)(x) &= \sum_{i=1}^{k_n} \left( \prod_{j=1}^{m_i} \left( \sum_{t = 1}^n a_{(i,j,t)}x^{p^t} \right) \right)\\
		&= \sum_{B \in \pow(\{ 1,\ldots,n \})} \left( a_B \prod_{j \in B} x^{p^{n-(n-j)}} \right)
	\end{align*}
	where $\pow(\cdot)$ denotes the power set and $a_B \in \ff_p$ for every $B \in \pow(\{ 1,\ldots,n \})$. Recall from \Cref{propsFF} \cref{expo} that $x^{p^{m+n}} = x^{p^m}$; thus there exists some $\ell \geq 1$ for which
	\[
		\deg(f_n) \leq p^{n-1} + \ldots + p^{n-\ell} \leq (p-1)\left(p^{n-1} + \ldots + p^{n-\ell} \right) = p^n - p^{n-\ell} \qedhere
	\]
\end{proof}

We now combine the field-theoretic tools above to prove the main theorem of this section.

\begin{theorem}\label{thm:sac2}
	Let $q \in (0,1/2)$ and $|\Sigma| = p$ for some prime $p$. Then $L_q \not\in \psac$ and~$\sstar \setminus L_q \not\in \psac$.
\end{theorem}

\begin{proof}
Suppose some circuit recognises $L_q$. By \Cref{isoProp}, there exists a family of polynomials $(f_n)$ and $\ell \in \nn$ for which $x \in L_q$ if and only if $(f_n \circ \phi_{n})(x) \neq 0$ and~$\deg(f_n) \leq p^n - p^{n-\ell}$. So, the number of roots of~$f_n$---and, hence, the number of words not in~$L_q$---is bounded above by~$p^n - p^{n-\ell}$, so the cardinality of~$L_q$ is in~$\Omega(p^n)$, contradicting \Cref{usecor27}.

For~$\sstar \setminus L_q$, note that the circuit can be augmented by a constant number of layers to flip the output of~$f_n \circ \phi_{n}$ for any $n$ (note that~$\ran(f_n \circ \phi_n) \subseteq \ff_p$). If~$a_x = (f_n \circ \phi_{n})(x) \neq 0$ then use \Cref{propsFF} \Cref{expo} to see that $a_x^{p} = a_x$; thus, we have~$a_x^{p-1} = 1$, and so the polynomial~$\theta(x) = 1 - x^{p-1}$ satisfies
\[
	\theta(x) = 0 \iff a_x \neq 0.
\]
As~$p$ is fixed,~$\theta$ can be computed by a constant-depth circuit, which we may append to any~$\psac$-circuit recognising~$L_q$ to recognise~$\sstar \setminus L_q$. Since the former does not exist, neither does the latter.
\end{proof}

\subsection{Partial Generalisations to Non-prime-Cardinality Alphabets}\label{sec:last}

We provide a partial generalisation to non-prime-alphabets. Although our theorem reaches the same conclusion as \Cref{thm:sac2}, the generalisation is partial as we redefine the definition of~$\psac$-recognisability to make our arguments amenable to non-prime cardinality settings.

Fix $q \in (0,1/2)$ and an alphabet $\Sigma$ with $|\Sigma| = r$, where $r$ is not prime. Let $p > r$ be the smallest prime greater than $r$. Let~$\Sigma_p$ be an alphabet of cardinality~$p$ which contains~$\Sigma$. As before, identify~$\Sigma_p$ with $\ff_p$. We now work over~$\Sigma_p$.

\begin{definition}\label[definition]{newDfn}
	A language $L \subset \Sigma_r^*$ is \define{$\prsac$-recognisable} if it is $\psac$-recognisable over the field $\ff_p$ by a family of polynomials $(f_n)_{n \in \nn}$ for which
\[
	f_n \colon \ff_{p^n} \rightarrow \ff_p
\]
(as per \Cref{isoProp}) and for which the following conditions hold: for all $n \in \nn$,
	\begin{enumerate}
		\item $f_n(x) = 1$ if $x \in \Sigma_r^n \cap L_q$;
		\item $f_n(x) = 0$ if $x \in \Sigma^n_r \setminus L_q$;
		\item $f_n(x) \in \ff_p \setminus \{ 1 \}$ otherwise.
	\end{enumerate}
\end{definition}

We use this re-definition to code information about the language $\Sigma_r$ as it is embedded in~$\Sigma_p$. This renders \Cref{newDfn} more restrictive than \Cref{psacdfn}, so the following theorem is slightly weaker than its counterpart \Cref{thm:sac2}; the proofs are similar.

\begin{theorem}\label{finalThm}
	Let~$q \in (0,1/2)$ and~$|\Sigma| = r$ for some~$r \not\in \pp$. Then $L_q \not\in \prsac$ and $\sstar \setminus L_q \not\in \prsac$.
\end{theorem}

\section{Open Questions}\label{sec:open}

In this paper, we proved multiple results on the complexity of the measure $\cN$ via the proxy family of sets $\set{L_q}{q \in (0,1/2)}$. In particular, we showed that~$L_q$ is complicated from the viewpoint of pushdown automata (Theorems \ref{thm:Lq} and \ref{thm:compl} and \Cref{cor:ADextension}), and even certain Boolean circuits cannot recognise~$L_q$, nor its complement (\Cref{thm:sac1,thm:sac2}). We also proved the Shannon effect for~ $\cN{}$ (\Cref{thm:sh}). Pressing open questions pertain to refining these results on~$L_q$---and, ultimately, to understanding the measure~$\cN$ even better.

In \Cref{ss:c2}, and in \Cref{thm:sac2} in particular, we considered alphabets of prime cardinality and we gave a generalisation to non-prime-cardinality alphabets in \Cref{sec:last}. However, our proof uses a non-standard definition of~$\psac$. We wonder:

\begin{question}
Do the results from \Cref{thm:sac2} apply to arbitrary alphabets using the definition of~$\psac$ given in \Cref{psacdfn}?
\end{question} 

Finally, it is clear by definition that $\cN(x) \leq \cNu(x)$ for all $x \in \sstar$, for any finite alphabet $\Sigma$. This allows the extension of certain results from~$\cN$ to~$\cNu$. For instance, we note here that~\Cref{thm:richHigh,thm:compl,thm:sac1} also apply to~$\cNu$ since~$\cN(x) \le \cNu(x)$. However, whether the equality~$\cN(x) = \cNu(x)$ holds in general remains the cardinal open question to fully understand the impact of exactness in \Cref{dfn:eacc} compared to \Cref{dfn:unn}.

\begin{question}\label[question]{questionF}
Let $|\Sigma| = 2$. Does there exist $x \in \sstar$ for which $\cN(x) < \cNu(x)$?
\end{question}

Broadly, it is our hope that a better understanding of~$\cN$, via proxies such as the~$L_q$ sets or otherwise, will lead to a better understanding of~$\cNu$, and (non-deterministic) automatic complexity in general. Similar to the comparison between regular languages being characterised by both NFAs and DFAs, a proof of the fact that~$\cNu = \cN$ would simplify computing the non-deterministic automatic complexity while preserving its naturalness as a measure of complexity. Conversely,~$\cNu \neq \cN$ would show that non-deterministic automatic complexity is dependent on the actual computation path, which would also be of philosophical interest.


\bibliographystyle{plainurl}
\bibliography{automcomp_arXiv_references}

\end{document}